\documentclass[runningheads]{llncs}
\usepackage{enumerate}
\usepackage{tensor}
\usepackage{hyperref}
\usepackage{lscape}
\usepackage{esvect}
\usepackage{etex, xy} 
\usepackage{mathtools}
\usepackage{amsmath,amssymb,latexsym, amsfonts} 
\usepackage{qtree}
\usepackage{tikz}
\usepackage{tikz-qtree}
\usepackage{bussproofs}
\usepackage{float}
\usepackage{enumitem}

\usepackage{siunitx}

\newcommand{\comment}[1]{}
\newcommand{\onestepApp}{\Rightarrow_{\approx}}
\newcommand{\onestepFresh}{\Rightarrow_{\#}}
\newcommand{\fpe}{{\mbox{\scriptsize fp}_\approx}}
\newcommand{\nfpe}{{\mbox{\scriptsize nfp}_\approx}}

\newcommand{\barr}[1]{\ensuremath{\overline{#1}}}
\newcommand{\epc}{{\tt epc}}
\newcommand{\Dta}{\SI{}{\Delta}}

\newdimen\proofrulebreadth \proofrulebreadth=.05em
\newdimen\proofdotseparation \proofdotseparation=1.25ex
\newdimen\proofrulebaseline \proofrulebaseline=2ex
\newdimen\proofrulebaseline \proofrulebaseline=1.7ex
\newcount\proofdotnumber \proofdotnumber=3
\let\then\relax
\def\hfi{\hskip0pt plus.0001fil}
\mathchardef\squigto="3A3B
\newif\ifinsideprooftree\insideprooftreefalse
\newif\ifonleftofproofrule\onleftofproofrulefalse
\newif\ifproofdots\proofdotsfalse
\newif\ifdoubleproof\doubleprooffalse
\let\wereinproofbit\relax
\newdimen\shortenproofleft
\newdimen\shortenproofright
\newdimen\proofbelowshift
\newbox\proofabove
\newbox\proofbelow
\newbox\proofrulename
\def\shiftproofbelow{\let\next\relax\afterassignment\setshiftproofbelow\dimen0 }
\def\shiftproofbelowneg{\def\next{\multiply\dimen0 by-1 }%
\afterassignment\setshiftproofbelow\dimen0 }
\def\setshiftproofbelow{\next\proofbelowshift=\dimen0 }
\def\setproofrulebreadth{\proofrulebreadth}

\def\prooftree{
\ifnum  \lastpenalty=1
\then   \unpenalty
\else   \onleftofproofrulefalse
\fi
\ifonleftofproofrule
\else   \ifinsideprooftree
        \then   \hskip.5em plus1fil
        \fi
\fi
\bgroup
\setbox\proofbelow=\hbox{}\setbox\proofrulename=\hbox{}%
\let\justifies\proofover\let\leadsto\proofoverdots\let\Justifies\proofoverdbl
\let\using\proofusing\let\[\prooftree
\ifinsideprooftree\let\]\endprooftree\fi
\proofdotsfalse\doubleprooffalse
\let\thickness\setproofrulebreadth
\let\shiftright\shiftproofbelow \let\shift\shiftproofbelow
\let\shiftleft\shiftproofbelowneg
\let\ifwasinsideprooftree\ifinsideprooftree
\insideprooftreetrue
\setbox\proofabove=\hbox\bgroup$\displaystyle 
\let\wereinproofbit\prooftree
\shortenproofleft=0pt \shortenproofright=0pt \proofbelowshift=0pt
\onleftofproofruletrue\penalty1
}

\def\eproofbit{
\ifx    \wereinproofbit\prooftree
\then   \ifcase \lastpenalty
        \then   \shortenproofright=0pt  
        \or     \unpenalty\hfil         
        \or     \unpenalty\unskip       
        \else   \shortenproofright=0pt  
        \fi
\fi
\global\dimen0=\shortenproofleft
\global\dimen1=\shortenproofright
\global\dimen2=\proofrulebreadth
\global\dimen3=\proofbelowshift
\global\dimen4=\proofdotseparation
\global\count255=\proofdotnumber
$\egroup  
\shortenproofleft=\dimen0
\shortenproofright=\dimen1
\proofrulebreadth=\dimen2
\proofbelowshift=\dimen3
\proofdotseparation=\dimen4
\proofdotnumber=\count255
}

\def\proofover{
\eproofbit 
\setbox\proofbelow=\hbox\bgroup 
\let\wereinproofbit\proofover
$\displaystyle
}%
\def\proofoverdbl{
\eproofbit 
\doubleprooftrue
\setbox\proofbelow=\hbox\bgroup 
\let\wereinproofbit\proofoverdbl
$\displaystyle
}%
\def\proofoverdots{
\eproofbit 
\proofdotstrue
\setbox\proofbelow=\hbox\bgroup 
\let\wereinproofbit\proofoverdots
$\displaystyle
}%
\def\proofusing{
\eproofbit 
\setbox\proofrulename=\hbox\bgroup 
\let\wereinproofbit\proofusing
\kern0.3em$
}

\def\endprooftree{
\eproofbit 
  \dimen5 =0pt
\dimen0=\wd\proofabove \advance\dimen0-\shortenproofleft
\advance\dimen0-\shortenproofright
\dimen1=.5\dimen0 \advance\dimen1-.5\wd\proofbelow
\dimen4=\dimen1
\advance\dimen1\proofbelowshift \advance\dimen4-\proofbelowshift
\ifdim  \dimen1<0pt
\then   \advance\shortenproofleft\dimen1
        \advance\dimen0-\dimen1
        \dimen1=0pt
        \ifdim  \shortenproofleft<0pt
        \then   \setbox\proofabove=\hbox{%
                        \kern-\shortenproofleft\unhbox\proofabove}%
                \shortenproofleft=0pt
        \fi
\fi
\ifdim  \dimen4<0pt
\then   \advance\shortenproofright\dimen4
        \advance\dimen0-\dimen4
        \dimen4=0pt
\fi
\ifdim  \shortenproofright<\wd\proofrulename
\then   \shortenproofright=\wd\proofrulename
\fi
\dimen2=\shortenproofleft \advance\dimen2 by\dimen1
\dimen3=\shortenproofright\advance\dimen3 by\dimen4
\ifproofdots
\then
        \dimen6=\shortenproofleft \advance\dimen6 .5\dimen0
        \setbox1=\vbox to\proofdotseparation{\vss\hbox{$\cdot$}\vss}%
        \setbox0=\hbox{%
                \advance\dimen6-.5\wd1
                \kern\dimen6
                $\vcenter to\proofdotnumber\proofdotseparation
                        {\leaders\box1\vfill}$%
                \unhbox\proofrulename}%
\else   \dimen6=\fontdimen22\the\textfont2 
        \dimen7=\dimen6
        \advance\dimen6by.5\proofrulebreadth
        \advance\dimen7by-.5\proofrulebreadth
        \setbox0=\hbox{%
                \kern\shortenproofleft
                \ifdoubleproof
                \then   \hbox to\dimen0{%
                        $\mathsurround0pt\mathord=\mkern-6mu%
                        \cleaders\hbox{$\mkern-2mu=\mkern-2mu$}\hfill
                        \mkern-6mu\mathord=$}%
                \else   \vrule height\dimen6 depth-\dimen7 width\dimen0
                \fi
                \unhbox\proofrulename}%
        \ht0=\dimen6 \dp0=-\dimen7
\fi
\let\doll\relax
\ifwasinsideprooftree
\then   \let\VBOX\vbox
\else   \ifmmode\else$\let\doll=$\fi
        \let\VBOX\vcenter
\fi
\VBOX   {\baselineskip\proofrulebaseline \lineskip.2ex
        \expandafter\lineskiplimit\ifproofdots0ex\else-0.6ex\fi
        \hbox   spread\dimen5   {\hfi\unhbox\proofabove\hfi}%
        \hbox{\box0}%
        \hbox   {\kern\dimen2 \box\proofbelow}}\doll%
\global\dimen2=\dimen2
\global\dimen3=\dimen3
\egroup 
\ifonleftofproofrule
\then   \shortenproofleft=\dimen2
\fi
\shortenproofright=\dimen3
\onleftofproofrulefalse
\ifinsideprooftree
\then   \hskip.5em plus 1fil \penalty2
\fi
}
\newcommand{\appAC}{\approx_{\{\alpha, C\}}}

\newcommand{\rulefont}[1]{\ensuremath{\mathbf{(#1)}}}

\makeatletter
\def\moverlay{\mathpalette\mov@rlay}
\def\mov@rlay#1#2{\leavevmode\vtop{%
   \baselineskip\z@skip \lineskiplimit-\maxdimen
   \ialign{\hfil$\m@th#1##$\hfil\cr#2\crcr}}}
\newcommand{\charfusion}[3][\mathord]{
    #1{\ifx#1\mathop\vphantom{#2}\fi
        \mathpalette\mov@rlay{#2\cr#3}
      }
    \ifx#1\mathop\expandafter\displaylimits\fi}
\makeatother

\newcommand{\pair}[2]{\ensuremath{\langle #1, #2\rangle }}
\newcommand{\triple}[3]{\ensuremath\langle #1, #2, #3 \rangle}

\title{Nominal C-Unification}
\titlerunning{Nominal C-Unification}

\author{Mauricio Ayala-Rinc\'on\inst{1}\thanks{Work supported by
    the Brazilian agencies FAPDF (DE 193.001.369/2016), CAPES (Proc. 88881.132034/2016-01, 2nd author) and CNPq (PQ
    307009/2013, 1st author).}, Washington de Carvalho-Segundo\inst{1},
Maribel Fern\'andez\inst{2} \and Daniele Nantes-Sobrinho\inst{1}}

\authorrunning{M. Ayala-Rinc\'on, W. de Carvalho, M. Fern\'andez \and
  D. Nantes-Sobrinho}
\institute{Depts. de Matem\'atica e Ci\^encia da Computa\c{c}\~ao, 
  Universidade de Bras\'ilia, Brazil
\email{ayala@unb.br, wtonribeiro@gmail.com,
 maribel.fernandez@kcl.ac.uk, d.n.sobrinho@mat.unb.br}
\and
  Department of Informatics,
  King's College London, UK}
  
\begin{document}
\maketitle

\begin{abstract}
Nominal unification is an extension of first-order unification that
takes into account the $\alpha$-equivalence relation generated by
binding operators, following the nominal approach. 
We propose a sound and complete procedure for nominal unification with
commutative operators,  or nominal C-unification for short, which has
been formalised in Coq.  The procedure transforms nominal C-unification
problems  into simpler (finite families) of \emph{fixpoint} problems, whose solutions can
be generated by algebraic techniques on combinatorics of permutations.  
\end{abstract}

\vspace{-3mm}
\section{Introduction}
\vspace{-3mm}

Unification, where the goal is to solve equations between first-order
terms, is a key notion in logic programming systems, type inference
algorithms, protocol analysis tools, theorem provers, etc.  Solutions
to unification problems are represented by substitutions that map
variables ($X, Y, \dots$) to terms.

When terms include binding operators, a more general notion of
unification is needed: unification modulo $\alpha$-equivalence.  In
this paper, we follow the nominal approach to the specification of
binding operators~\cite{Gabbay2002a,Urban2004,pitts2013nominal}, where
the syntax of terms includes, in addition to variables, also
\emph{atoms} ($a, b, \dots)$, which can be abstracted, and
$\alpha$-equivalence is axiomatised by means of a \emph{freshness
  relation} $a\#t$ and \emph{name-swappings} $(a\,b)$. For example,
the first-order logic formula $\forall a. a \geq 0$ can be written as
a nominal term $\forall([a]geq(a,0))$, using function symbols
$\forall$ and $geq$ and an abstracted atom $a$.  Nominal
unification~\cite{Urban2004} is the problem of solving equations
between nominal terms modulo $\alpha$-equivalence; it is a decidable
problem and efficient nominal unification algorithms are
available~\cite{Calves2011,Calves2010a,Levy2010}, that compute
solutions consisting of \emph{freshness contexts} (containing
freshness constraints of the form $a\# X$) and substitutions.

In many applications, operators obey equational axioms. Nominal
reasoning and unification have been extended to deal with equational
theories presented by rewrite rules (see, e.g.,
\cite{Fernandez2004,Fernandez2010,ARFNantes2015}) or defined by
equational axioms (see, e.g., \cite{Clouston2007,GabbayM09}). The case
of associative and commutative nominal theories was considered
in~\cite{Ayala-Rincon2016}, where a parametric
$\{\alpha, AC\}$-equivalence relation was formalised in Coq. However,
only equational deduction was considered (not unification).  In this
paper, we study nominal C-unification.

\noindent
\textbf{Contributions:} We present a nominal C-unification algorithm,
based on a set of \emph{simplification rules}, which transforms a
given \emph{nominal C-unification problem} $\langle \Dta, Q\rangle$,
where $\Dta$ is a freshness context and $Q$ a set of freshness
constraints and equations, respectively of the form $a \#_? s$ and
$s \approx_? t$, into a finite set of triples of the form
$\langle \nabla,\sigma, P\rangle$, consisting of a freshness context
$\nabla$, a substitution $\sigma$ and a set of fixpoint equations $P$,
of the form $\pi.X \approx_?  X$. The simplifications are based on the
deduction rules for freshness and $\alpha$-C-equivalence (denoted as
$\appAC$).

The role of fixpoint equations in nominal C-unification is
  tricky: while in standard nominal unification~\cite{Urban2004},
  solving a fixpoint equation of the form $(a \ b).X \approx_? X$
  reduces to checking whether the constraints $a\# X, b \# X$ ($a$ and
  $b$ fresh in $X$) are satisfied, and in this case the solution is
  the {\em identity} substitution, in nominal C-unification, for $*$
  and $+$ commutative operators, one can have additional combinatory
  solutions of the form
  $\{X/ a + b \}, \{X/(a + b)* \ldots * (a + b)\}$, $\{X/f(a)+f(b)\}$,
  etc.
  We show that in general there is no finitary representation of
  solutions using only freshness contexts and substitutions, hence a
  nominal C-unification problem may have a potentially infinite set of
  independent most general unifiers (unlike standard C-unification,
  which is well-known to be finitary).

  We adapt the proof of NP-completeness of syntactic C-unification to show that nominal
  C-unification is NP-complete as well.  

 The simplification rules were formalised to be sound and
  complete in Coq. The formalisation and a full version of the present
  paper with all details of the proofs are available at
  \href{http://mat.unb.br/~ayala/publications.html}{\tt \color{blue}http://ayala.mat.unb.br/publications.html.} 

\noindent\textbf{Related work}: 
To  generate the set of combinatorial solutions for fixpoint
  equations  we can use an enumeration procedure given
  in~\cite{Ayala2017}, which is based on the combinatorics of
  permutations.  By combining the simplification and enumeration
  methods, we obtain a nominal C-unification procedure in two phases:
  a \emph{simplification phase}, described in this paper, which outputs a finite set of most
  general solutions that may include fixpoint constraints, and a
  \emph{generation phase}, which eliminates the fixpoint constraints
  according to~\cite{Ayala2017}. 

Several extensions of the nominal unification
algorithm have been defined, in addition to the equational extensions
already mentioned.

An algorithm for nominal unification of higher-order expressions with
recursive {\em let} was proposed in~\cite{Kutsia2016}; as for nominal
C-unification, fixpoint equations are obtained in the
process. Using the techniques in \cite{Ayala2017}, it is
  possible to proceed further  and generate the 
  combinatorial solutions of fixpoint equations.

Recently, Aoto and Kikuchi~\cite{Aoto2016a} proposed a rule-based
procedure for nominal equivariant unication~\cite{Cheney2010}, an
extension of nominal unification that is useful in confluence analysis
of nominal rewriting systems~\cite{Aoto2016,Fernandez2007}.

Furthermore, several formalisations and implementations of the nominal
unification algorithm are available.  For example, formalisations of
its soundness and completeness were developed by Urban et
al~\cite{Urban2004,Urban2010}, Ayala-Rinc\'on et
al~\cite{oliveira2015completeness}, and Kumar and
Norrish~\cite{Kumar2010} using, respectively, the proof assistants
Isabelle/HOL, PVS and HOL4.  An implementation in Maude using term
graphs~\cite{Calves2007} is also available.  Urban and Cheney used a
nominal unification algorithm to develop a Prolog-like language called
$\alpha$-Prolog~\cite{Cheney2003}.
Our formalisation of nominal C-unification is based on the
  formalisation of equivalence modulo $\{\alpha, AC\}$ presented
  in~\cite{Ayala-Rincon2016}. The representations of permutations and
  terms are similar, but here we deal also with substitutions and
  unification rules, and prove soundness and completeness of the
  unification algorithm.

 Non nominal reasoning modulo equational theories has been
  subject of formalisations. For instance, in \cite{Nipkow89}, Nipkow
  presented a set of Isabelle/HOL tactics for reasoning modulo A, C and AC; Braibant and Pous~\cite{BraibantP11} designed a plugin
  for Coq, with an underlying AC-matching algorithm, that extends the
  system tactic {\tt rewrite} to deal with AC function symbols; also,
  Contejean \cite{Contejean04} formalised in Coq the correction of an
  AC-matching algorithm implemented in 
C\emph{i}ME.  

Syntactic unification with commutative operators 
is an NP-complete problem and its solutions can be
finitely generated~\cite{Kapur1987,Siekmann1979}. 
Since C-unification problems are a particular case of
nominal C-unification problems, our simplification algorithm, checked
in Coq, is also a formalisation of the C-unification algorithm.

\noindent\textbf{Organisation}: 
Section~\ref{sec:background} presents basic concepts and notations. 
Section~\ref{sec:unificationalgorithm} introduces the formalised
equational and freshness inference rules for nominal
C-unification, and also discusses NP-completeness; 
Section~\ref{sec:complete_set_for_trees} shows how single fixpoint
equations are solved giving rise to infinite independent solutions and
also we briefly explain how fixpoint solutions are combined
in~\cite{Ayala2017} in order to produce 
solutions. Section~\ref{sec:conclusions} concludes and describes
future work. 

\vspace{-3mm}
\section{Background}
\label{sec:background}
\vspace{-3mm}

Consider countable disjoint sets of variables ${\cal X} := \{X, Y, Z,
\cdots\}$ and atoms ${\cal A} := \{a, b, c, \cdots\}$. 
A \emph{permutation} $\pi$ is a bijection on $\mathcal{A}$ with a finite
\emph{domain}, where the domain (i.e., the \emph{support})  of $\pi$ is the set 
$dom(\pi) := \{a \in {\cal A}\;|\;\pi\cdot a \neq a\}$. 
The inverse of $\pi$ is denoted by $\pi^{-1}$.
Permutations can be represented by  lists of  {\em
  swappings}, which are  pairs of different atoms $(a\,b)$; hence a
{\em permutation} $\pi$ is a finite list of the form
$(a_1\,b_1)::\ldots ::(a_n\,b_n):: \mbox{\em nil}$, where \emph{nil} 
denotes the identity permutation;  concatenation is denoted
by $\oplus$ and, when no confusion may arise, $::$ and $nil$ are omitted.   
 We follow Gabbay's permutative convention:
Atoms  differ on their names, so for atoms $a$ and $b$ the expression
  $a\neq b$ is redundant. Also, $(a\,b)$ and $(b\,a)$
    represent the same swapping.

We will assume as in \cite{Ayala-Rincon2016}
countable sets of function symbols with different equational
properties such as associativity, commutativity, idempotence,
etc. Function symbols have superscripts that indicate their
equational properties; thus, $f^C_k$ will denote the $k^{th}$ function
symbol that is commutative and $f^\emptyset_j$ the $j^{th}$ function
symbol without any equational property. 
 
\emph{Nominal terms} are generated by the following
  grammar: \\[1mm]
\mbox{}\hspace{3.2cm}$s, t := \langle\rangle \,\,|\,\, \bar{a} \,\,|\,\, [a]t \,\, | \,\,
\langle s, t\rangle \,\,|\,\, f^E_k\,t \,\,|\,\, \pi.X$\\[1mm]
$\langle\rangle$ denotes the \emph{unit} (that is the empty tuple), $\bar{a}$ denotes an
\emph{atom term},
$[a]t$ denotes an \emph{abstraction} of the atom $a$ over the term $t$, $\langle s,
t\rangle$ denotes a \emph{pair}, $f^E_k\,t$ the  \emph{application} of
$f^E_k$ to $t$ and, $\pi.X$ a \emph{moderated variable} or \emph{suspension}. 
 Suspensions of the form ${\tt id} . X$ will be represented just by
 $X$.

The set of variables occurring in a term $t$ will be denoted as $Var(t)$. 
This notation extends to a set $S$ of terms in the natural way:
$Var(S) = \bigcup_{t \in S} Var(t)$.  
As usual, $|\,\_ \,|$ will be used to denote the cardinality
  of sets as well as to denote the size or number of symbols occurring
  in a given term. 


\begin{definition}[Permutation action] 
The action of a permutation on atoms is defined as: 
$\mbox{\it nil}\cdot a := a$; $(b\,c)::\pi \cdot a := \pi\cdot a$;
and, $(b\,c)::\pi \cdot b := \pi\cdot c$.  The action of a permutation
on terms is defined recursively as: \\[1mm]
$\mbox{}\hspace{4mm}\begin{array}{l@{\;\;:=\;\;}l@{\hspace{7mm}}l@{\;\;:=\;\;}l@{\hspace{7mm}}l@{\;\;:=\;\;}l}
     \pi\cdot \langle\rangle & \langle\rangle  &
  \pi\cdot\langle u, v\rangle & \langle \pi\cdot u, \pi\cdot v\rangle & 
    \pi\cdot f_k^E\,t & f_k^E\,(\pi\cdot  t) \\
     \pi\cdot \overline{a} & \overline{\pi \cdot a} &
\pi\cdot ([a]t) & [\pi\cdot a](\pi \cdot t) & 
\pi\cdot (\pi'\,.\,X) & (\pi'\oplus\pi)\,.\,X
   \end{array}$
\end{definition}

Notice that according to the definition of the action of a permutation
over atoms, the composition of permutations $\pi$ and $\pi'$, usually
denoted as $\pi\circ\pi'$, corresponds to the append $\pi'\oplus\pi$. 
Also notice that $\pi'\oplus\pi\cdot t = \pi\cdot(\pi'\cdot t$).
%
The \emph{difference set} between two permutations $\pi$ and $\pi'$ 
is the set of atoms where the action of $\pi$ and $\pi'$ differs:
$ds(\pi,\pi') := \{a \in {\cal A} \;|\; \pi\cdot a
\neq \pi'\cdot a\}$.

A \emph{substitution} $\sigma$ is a mapping from variables to terms
such that its \emph{domain}, $dom(\sigma) :=\{X \mid
X\neq X\sigma\}$, is finite.   For $X\in dom(\sigma)$, $X\sigma$ is called the
\emph{image} of $X$.   Define the \emph{image} of $\sigma$ as $im(\sigma) :=
\{X\sigma\mid X\in dom(\sigma)\}$. 
Let $dom(\sigma) = \{X_1, \cdots, X_n\}$, then $\sigma$ can
be represented as a set of \emph{bindings}  in the form $\{X_1/t_1,
\cdots, X_n/t_n\}$, where $X_i\sigma = t_i$,  for $1 \leq i \leq n$. 


\begin{definition}[Substitution action]\label{def:subst_action}
The \emph{action of a substitution} $\sigma$ on a term $t$, denoted 
$t\sigma$, is defined recursively as follows:\\[1mm]
 $\mbox{}\hspace{8mm}\begin{array}{lcl@{\hspace{1.5cm}}lcl@{\hspace{1.5cm}}lcl}
                           \langle\rangle\sigma & := & \langle\rangle &
                           \overline{a}\sigma & := & \overline{a}  &
                            (f_k^E\,t)\sigma & := & f_k^E\,t\sigma \\
                           \langle s, t\rangle\sigma & := &\langle s\sigma, t\sigma\rangle &
                           ([a]t)\sigma & := & [a]t\sigma &
                           (\pi.X)\sigma & := & \pi \cdot X\sigma
   \end{array}$
\end{definition}

The following result can be proved by induction on the structure of terms. 

\begin{lemma}[Substitutions and Permutations Commute]\label{lem:perm_subst_commutes}
$(\pi \cdot t)\sigma = \pi \cdot (t\sigma)$
\end{lemma}

 The inference rules defining freshness and $\alpha$-equivalence are given in Fig.
\ref{fig:freshness_relation} and \ref{fig:alpha_equivalence}.  
The symbols $\nabla$ and $\Dta$ are used to
denote \emph{freshness contexts} that are sets of  constraints of
the form $a\#X$, meaning that the atom 
$a$ is fresh in $X$. The domain of a freshness context $dom(\Dta)$
is the set of atoms appearing in it;  $\Dta|_X$ denotes the
restriction of $\Dta$ to the freshness constraints on $X$: $\{a\#X\;|\; a\#X\in \Dta\}$.
The  rules in
Fig. \ref{fig:freshness_relation} are used to check if an atom $a$ is 
fresh in a nominal term $t$ under a freshness context $\nabla$, also
denoted as $\nabla\vdash a\#t$. 
The  rules in Fig. \ref{fig:alpha_equivalence} are used to
check if two nominal terms $s$ and $t$ are $\alpha$-equivalent under some
freshness context $\nabla$, written as $\nabla\vdash
s\approx_\alpha t$. These rules use the inference system
for freshness constraints: specifically freshness constraints are used
in rule $\rulefont{\approx_\alpha
  [ab]}$. 
   
\begin{example}
Let $\sigma = \{X/[a]a\}$. Verify that 
 $\pair{(a \ b ). X}{f(e)}\sigma \approx_{\alpha} \pair{X}{f(e)}\sigma$. 
\end{example}

By $dom(\pi)\#X$ and $ds(\pi,\pi')\#X$ we abbreviate  
the sets $\{a\#X\;|\;a \in dom(\pi)\}$ and $\{a\#X\;|\;a \in
ds(\pi,\pi')\}$, respectively.

\begin{figure}[http]
\vspace{-5mm}
\[
\boxed{
\begin{array}{cccc}
\begin{prooftree}
\justifies \nabla \vdash a\,\#\,\langle\rangle
\using \rulefont{\# \langle\rangle}
\end{prooftree}
&
\begin{prooftree}
\justifies \nabla\vdash a \,\#\, \overline{b}
\using \rulefont{\#\, atom}
\end{prooftree}
&
\begin{prooftree}
\nabla \vdash a \,\#\, t
\justifies \nabla \vdash a \,\#\, f_k^E\,t 
\using \rulefont{\#\, app}
\end{prooftree}
&
\begin{prooftree}
\justifies \nabla\vdash a \,\#\, [a]t
\using \rulefont{\#\,a[a]}
\end{prooftree}
\\[3ex]
\multicolumn{4}{c}
{\begin{array}{ccc}
\begin{prooftree}
\nabla \vdash a \,\#\, t
\justifies \nabla \vdash a \,\#\, [b]t 
\using \rulefont{\#\, a[b]}
\end{prooftree}
&
\begin{prooftree}
(\pi^{-1}\cdot a\#X) \in \nabla
\justifies \nabla\vdash a \,\#\, \pi . X
\using \rulefont{\#\,var}
\end{prooftree}
&
\begin{prooftree}
\nabla\vdash a\,\#\, s \,\,\, \nabla\vdash a\,\#\, t
\justifies \nabla\vdash a \,\#\, \langle s, t\rangle
\using \rulefont{\#\,pair}
\end{prooftree}
\end{array}}
\end{array}
}
\]
\vspace{-5mm}
\caption{Rules for the freshness relation} 
\label{fig:freshness_relation}
\vspace{-2mm}
\end{figure}

\begin{figure}[http]
\[
\boxed{
\begin{array}{ccc}
\begin{prooftree}
\justifies \nabla\vdash\langle\rangle\approx_{\alpha}\langle\rangle
\using \rulefont{\approx_\alpha \langle\rangle}
\end{prooftree}
&
\begin{prooftree}
\justifies \nabla\vdash \overline{a} \approx_\alpha \overline{a}
\using \rulefont{\approx_\alpha atom}
\end{prooftree}
&
\begin{prooftree}
\nabla \vdash s \approx_\alpha t
\justifies \nabla \vdash  f_k^E\,s \approx_\alpha f_k^E\,t 
\using \rulefont{\approx_\alpha app}
\end{prooftree}
\\[3ex]
\multicolumn{3}{c}
{\begin{array}{cc}
\begin{prooftree}
\nabla \vdash s \approx_\alpha t
\justifies \nabla\vdash [a]s \approx_\alpha [a]t
\using \rulefont{\approx_\alpha [aa]}
\end{prooftree}
&
\begin{prooftree}
\nabla \vdash s \approx_\alpha (a\,b)\cdot t \;\; \nabla \vdash a \,\#\, t
\justifies \nabla \vdash [a]s \approx_\alpha [b]t 
\using \rulefont{\approx_\alpha [ab]}
\end{prooftree}
\end{array}}
\\[3ex]
\multicolumn{3}{c}
{\begin{array}{cc}
\begin{prooftree}
ds(\pi,\pi')\# X \subseteq \nabla
\justifies \nabla\vdash \pi.X \approx_\alpha \pi'.X
\using \rulefont{\approx_\alpha var}
\end{prooftree}
&
\begin{prooftree}
\nabla\vdash s_0 \approx_\alpha t_0 \;\; \nabla\vdash s_1 \approx_\alpha t_1
\justifies \nabla\vdash \langle s_0, s_1\rangle \approx_\alpha \langle t_0, t_1\rangle
\using \rulefont{\approx_\alpha pair}
\end{prooftree}
\end{array}}
\end{array}
}
\]
\vspace{-3mm}
\caption{Rules for the relation $\approx_\alpha$} 
\label{fig:alpha_equivalence}
\vspace{-6mm}
\end{figure}

Key properties of the nominal freshness and $\alpha$-equivalence
relations have been extensively explored in previous
works~\cite{Ayala-Rincon2016,oliveira2015completeness,Urban2010,Urban2004}.

\subsection{The relation $\appAC $ as an extension of $\approx_\alpha$}

In \cite{Ayala-Rincon2016},  the relation
$\approx_\alpha$ was extended to deal with associative and commutative 
theories. Here we will consider $\alpha$-equivalence modulo
commutativity, denoted as $\appAC $. This means that some function
symbols in our syntax are commutative, and therefore the rule for
function application  $\rulefont{\approx_\alpha app}$ in
Fig. \ref{fig:alpha_equivalence} 
should be replaced by the  rules in Fig.~\ref{fig:alpha_C_equivalence}.

\begin{figure}[http]
\vspace{-7mm}
\[
\boxed{
\begin{array}{c}
\begin{prooftree}
\nabla \vdash s \appAC  t
\justifies \nabla \vdash  f_k^E\,s \appAC  f_k^E\,t 
\using,\;\;E \neq C\mbox{ or both } s \mbox{ and } t
\mbox{ are not pairs}\;\;\rulefont{\appAC  app}
\end{prooftree}
\\[3ex]
\begin{prooftree}
\nabla \vdash s_0 \appAC  t_i, \;\;\nabla \vdash s_1 \appAC  t_{(i+1)\,mod\, 2} 
\using, \;\;i = 0, 1\;\; \rulefont{\appAC  C} 
\justifies
\nabla \vdash f^C_k\,\langle s_0, s_1\rangle \appAC  f^C_k\,\langle t_0, t_1\rangle
\end{prooftree}
\end{array}
}
\]
\vspace{-3mm}
\caption{Additional rules for $\{\alpha,C\}$-equivalence} 
\label{fig:alpha_C_equivalence}
\vspace{-7mm}
\end{figure}

The following properties for $\appAC$ were formalised as simple
 adaptations of the formalisations given in~\cite{Ayala-Rincon2016}
 for $\approx_\alpha$.

\begin{lemma}[Inversion]\label{lem:eq_reverse}
  The inference rules
   of $\appAC$ are \emph{invertible}.
\end{lemma}
This means, for instance, that for rules \rulefont{\approx_\alpha [ab]}  
one has $\nabla \vdash [a]s \appAC  [b]t$ implies  
   $\nabla \vdash s \appAC  (a\,b)\cdot t$ and  
   $\nabla \vdash a \,\#\, t$; and for \rulefont{\appAC  app}, $\nabla
   \vdash f^C_k\,\langle s_0, s_1 \rangle \appAC      f^C_k\,\langle
   t_0, t_1 \rangle$ implies  $\nabla \vdash s_0 \appAC  t_0$   and
   $\nabla \vdash s_1 \appAC  t_1$, or  $\nabla \vdash s_0 \appAC
   t_1$    and $\nabla \vdash s_1 \appAC  t_0$. 

\begin{lemma}[Freshness preservation] 
\label{lem:fresh_preservation} 
  If $\nabla \vdash a\,\#\,s$ and $\nabla \vdash s \appAC  t$ then
  $\nabla \vdash a\,\#\,t$.
\end{lemma}

\begin{lemma}[Intermediate transitivity for $\appAC $
  with $\approx_\alpha$]
If $\nabla \vdash s \appAC  t$ and $\nabla \vdash t \approx_\alpha u$ then $\nabla \vdash s \appAC  u$.
\end{lemma}

\begin{lemma}[Equivariance] 
\label{lem:equivariance}
 $\nabla \vdash \pi\cdot s \appAC  \pi\cdot t$ whenever $\nabla \vdash s \appAC  t$.
\end{lemma}

\begin{lemma}[Equivalence] 
\label{lem:equivalence}
$\_ \vdash \_ \appAC  \_$ is an equivalence relation.
\end{lemma}

\begin{remark}
According to the grammar for nominal terms, function symbols have no fixed arity: 
any function symbol can apply to any term.  Despite this, in the syntax
of our Coq formalisation commutative symbols apply only to  tuples.
\end{remark}

 
\vspace{-5mm}
\section{A nominal C-unification algorithm}\label{sec:unificationalgorithm}
\vspace{-3mm}

Inference rules are given that transform a nominal C-unification
problem into a finite family of problems that consist 
exclusively of fixpoint equations of the form $\pi.X \approx_? X$, 
together with a substitution and a set of freshness constraints.

\begin{definition}[Unification problem]\label{def:uproblem}
A {\em unification problem} is a pair $\langle \nabla,
P\rangle$, where $\nabla$ is a \emph{freshness context} and $P$ is a
finite set of {\em equations} and \emph{freshness constraints} of
the form $s \approx_? t$ and  $a \#_? s$, respectively, where
$\approx_?$ is symmetric, $s$ and $t$ are terms and $a$ is an
atom.  Nominal terms in the equations preserve the syntactic
  restriction that commutative symbols are only applied to tuples. 
\end{definition}

Equations of the form $\pi.X \approx_? X$ are called  {\em   fixpoint} equations.
    Given $\langle \nabla, P\rangle$, by $P_\approx, P_\#, P_\fpe$ and
    $P_\nfpe$ we will resp. 
    denote the
  sets of equations, freshness constraints, fixpoint equations and non
  fixpoint equations in the set $P$.

\begin{example}\label{ex:0}
  Given the nominal unification problem
  $\mathcal{P}\!=\!\pair{\emptyset}{\{[a][b]X \approx_? [b][a]X\}}$, the
  standard unification algorithm~\cite{Urban2004} reduces it
  to $\pair{\emptyset}{\{X \approx_? (a\,b).X\}}$, which gives the
  solution $\pair{\{a\#X, b\#X\}}{id}$. 
 However, we will see that infinite independent solutions
  are feasible when there is at least a commutative operator.
\end{example}

We design a nominal C-unification  algorithm using one set of transformation rules to deal with 
equations (Fig. \ref{fig:eq_rules_algorithm}) and
another set of rules  to deal with freshness constraints and contexts (Fig.
\ref{fig:fresh_rules_algorithm}). These rules act over triples of the form  
$\langle \nabla, \sigma, P\rangle$, where $\sigma$ is a substitution.
The triple that will be associated by default with a unification
problem $\langle \nabla, P\rangle$ is $\langle \nabla, id,  P\rangle$.
We will use calligraphic uppercase letters (e.g., $\cal P, Q, R,$
etc) to denote triples.

\begin{remark}     Let $\nabla$ and $\nabla'$ be freshness contexts and
  $\sigma$ and $\sigma'$ be substitutions. 
\begin{itemize}
  \item $\nabla' \vdash \nabla\sigma$ denotes that $\nabla' \vdash a\,\#\,X\sigma$ 
   holds for each $(a \# X) \in \nabla$, and
  \item $\nabla \vdash \sigma \approx \sigma'$ that $\nabla \vdash X\sigma \appAC  X\sigma'$
  for all $X$ (in $dom(\sigma) \cup dom(\sigma')$). 
\end{itemize}
\end{remark}
  
\begin{definition}[Solution for a triple or problem]\label{def:sol} 
A {\em solution} for a triple ${\cal P} = \langle \Dta, \delta, P\rangle$
is a pair $\langle \nabla, \sigma \rangle$, where 
the following conditions are satisfied: 

\vspace{1mm}
\begin{minipage}{.38\textwidth}
\begin{enumerate}[leftmargin=1mm]
 \item $\nabla \vdash \Dta\sigma$;
 \item  $\nabla \vdash a \,\#\,t\sigma$ if $a \#_? t \in P$;
\end{enumerate}
\end{minipage}\hspace{-.9cm}
\begin{minipage}{.61\textwidth}
\begin{enumerate}\setcounter{enumi}{2}
 \item  $\nabla \vdash
  s\sigma \appAC  t\sigma$ if $s \approx_? t \in P$;
 \item there is a substitution 
$\lambda$ such that $\nabla\vdash \delta  \lambda
   \!\approx\!  \sigma$.
\end{enumerate}
\end{minipage} 
\vspace{2mm}

A \emph{solution for a unification problem} $\langle \Dta, P\rangle$ is a
solution for the associated triple  $\langle \Dta, id, P\rangle$. 
The \emph{solution set} for a problem or triple ${\cal P}$  is denoted
by     $\mathcal{U}_C({\cal P})$. 
\end{definition}

\begin{definition}[More general solution and complete set of solutions]
  For $\langle \nabla, \sigma\rangle$ and
  $\langle \nabla', \sigma'\rangle$ in ${\cal U}_C({\cal P})$, we say
  that $\langle \nabla, \sigma\rangle$ is \emph{more general} than
  $\langle \nabla', \sigma'\rangle$, denoted
  $\langle \nabla, \sigma\rangle \preccurlyeq \langle \nabla',
  \sigma'\rangle$, if there exists a substitution $\lambda$ satisfying
  $\nabla' \vdash \sigma\lambda \approx \sigma'$ and
  $\nabla' \vdash \nabla\lambda$.  A subset $\cal{V}$ of
  ${\cal U}_C({\cal P})$ is said to be a \emph{complete set of
    solutions} of $\cal P$ if for all
  $\langle \nabla', \sigma'\rangle\in {\cal U}_C({\cal P})$, there
  exists $\langle \nabla, \sigma\rangle$ in $\cal V$ such that
  $\langle \nabla, \sigma\rangle \preccurlyeq \langle \nabla', \sigma'\rangle$.
\end{definition}

We will denote the set of variables occurring in the set $P$ of a
problem $\langle \nabla, P\rangle$ or triple ${\cal P}=\langle \nabla,\sigma,
P\rangle$ as $Var(P)$. We also will write $Var({\cal P})$ to denote
this set. 

\begin{figure}[http]
\vspace{-6mm}
\[
\boxed{
\begin{array}{cc}
\begin{prooftree}
\langle \nabla, \sigma, P \uplus \{s \approx_? s\}\rangle
\justifies \langle \nabla, \sigma, P\rangle
\using \rulefont{\approx_?refl}
\end{prooftree}
&
\begin{prooftree}
\langle \nabla, \sigma, P \uplus \{\langle s_1, t_1 \rangle \approx_? \langle s_2, t_2 \rangle\}\rangle
\justifies \langle \nabla, \sigma, P \cup \{s_1 \approx_? s_2, t_1 \approx_? t_2\}\rangle
\using \rulefont{\approx_? pair}
\end{prooftree}
\\[3ex]
\multicolumn{2}{c}
{\begin{prooftree}
\langle \nabla, \sigma, P \uplus \{f^E_k\,s \approx_? f^E_k\,t\}\rangle
\justifies 
\langle \nabla, \sigma, P \cup \{s \approx_? t\}\rangle
\using ,  \mbox{ if  $E \neq C$} \,\,\rulefont{\approx_? app} 
\end{prooftree}}
\\[3ex]
\multicolumn{2}{c}
{\begin{prooftree}  %
\langle \nabla, \sigma, P \uplus 
        \{f^C_k\,s  \approx_? f^C_k\,t \}\rangle
\using, \left\{\begin{array}{l}\mbox{where } s = \langle s_0, s_1\rangle
                \mbox{ and } t = \langle t_0, t_1\rangle\\ v =\langle
                t_i, t_{(i+1) \,mod\, 2}\rangle,  i = 0, 1 \end{array}\right\} \rulefont{\approx_?C} 
\justifies
\langle \nabla, \sigma, P \cup \{s \approx_? v \}\rangle
\end{prooftree}}
\\[3ex]
\begin{prooftree}
\langle \nabla, \sigma, P \uplus \{[a]s \approx_? [a]t\}\rangle
\justifies \langle \nabla, \sigma, P \cup \{s \approx_? t\}\rangle
\using \rulefont{\approx_?[aa]} 
\end{prooftree}
&
\begin{prooftree}
\langle \nabla, \sigma, P \uplus \{[a]s \approx_? [b]t\}\rangle
\justifies
\langle \nabla, \sigma, P \cup \{s \approx_? (a\,b)\,t, a \#_? t\}\rangle
\using  \rulefont{\approx_?[ab]}
\end{prooftree}
\\[3ex]
\multicolumn{2}{c}
{\begin{prooftree}
\langle \nabla, \sigma, P \uplus \{\pi . X  \approx_? t\}\rangle 
\,\, \mbox{ let } \sigma':= \sigma\{X / \pi^{-1}\cdot t\}
\justifies
\left\langle\begin{array}{l}
\nabla, \sigma', P\{X / \pi^{-1}\cdot t\} \,\,
\cup 
\begin{array}{l}
\bigcup\limits_{\substack{Y \in dom(\sigma'), \\a \# Y \in \nabla}}
\{a \#_? Y\sigma'\} 
\end{array}
\end{array}\right\rangle
\using , \mbox{ if } X \notin Var(t)  \,\,\rulefont{\approx_? inst}
\end{prooftree}}
\\[3ex]
\multicolumn{2}{c}
{\begin{prooftree}
\langle \nabla, \sigma, P \uplus \{\pi.X \approx_? \pi'.X\}\rangle
\justifies \langle \nabla, \sigma, P \cup  \{\pi\oplus(\pi')^{-1}.X \approx_? X\} \rangle
\using, \mbox{ if } \pi' \neq {\tt id} \,\,\rulefont{\approx_?inv}
\end{prooftree}}
\end{array}
}
\]\vspace{-3mm}
\caption{Reduction rules for equational problems} 
\label{fig:eq_rules_algorithm}
\vspace{-6mm}
\end{figure}

The unification algorithm proceeds by simplification.
Derivation with rules of Figs. \ref{fig:eq_rules_algorithm} and
\ref{fig:fresh_rules_algorithm} is respectively
denoted by $\onestepApp$ and $\onestepFresh$. Thus, $\langle
\nabla,\sigma,P\rangle \onestepApp \langle
\nabla,\sigma',P'\rangle$ means  that the second triple is obtained
from the first one by application of one rule. 
We will use the standard rewriting
nomenclature, e.g., we will
say that $\cal P$ is a \emph{normal form} or \emph{irreducible}
by $\onestepApp$, denoted by
\emph{$\onestepApp$-nf}, whenever there is no $\cal Q$ 
such that ${\cal P}\onestepApp{\cal Q}$; 
$\onestepApp^*$
and $\onestepApp^+$  denote respectively derivations in zero
or more and one or more applications of the rules in
Fig. \ref{fig:eq_rules_algorithm}.  

The only rule that can generate branches is
  $\rulefont{\approx_?C}$, which  is an abbreviation for two rules 
  providing  the different forms in which one can relate the arguments  $s$ and $t$ in an equation 
  $f^C_k\,s  \approx_? f^C_k\,t $ for a  commutative
  function symbol ($s$, $t$ are tuples, by the syntactic
  restriction in Definition \ref{def:uproblem}): either $\langle s_0, s_1\rangle \approx_?   \langle t_0,
  t_1\rangle$ or   $\langle s_0, s_1\rangle \approx_?   \langle t_1,
  t_0\rangle$.  
  
  The syntactic restriction on
  arguments of commutative symbols being only tuples,  is not crucial
  since any equation of the form $f_k^C  \pi. X \approx_? t$ can be
  translated into an  equation of  form $f_k^C  \langle \pi.X_1,
  \pi.X_2\rangle \approx_?   t$, where $X_1$ and $X_2$ are new
  variables and  $\nabla$ is   extended to $\nabla'$ in such a way
  that both $X_1$ and $X_2$   inherit all freshness constraints of $X$
  in $\nabla$:  $\nabla'   =\nabla \cup \{a\#X_i \;|\; i=1,2, \mbox{
    and } a\#X\in \nabla\}$.

In the rule $\rulefont{\approx_? inst}$ the inclusion of new
  constraints in the problem, given in   
\begin{minipage}{0,22\textwidth}
\scriptsize{$\displaystyle\bigcup_{\substack{Y \in dom(\sigma'), \\a \# Y \in \nabla}}
  \{a \#_? Y\sigma'\}$}
\end{minipage} 
\begin{minipage}{0,753\textwidth} \ is necessary to
guarantee that the new substitution $\sigma'$ is
{\em compatible} with the freshness context $\nabla$.
\end{minipage}

\begin{figure}[ht]
\vspace{-4mm}
\[
\boxed{
\begin{array}{cc}
\begin{prooftree}
\langle \nabla, \sigma, P \uplus \{a \#_? \langle\rangle\} \rangle
\justifies \langle \nabla, \sigma, P \rangle
\using \rulefont{\#_? \langle\rangle}
\end{prooftree}
&
\begin{prooftree}
\langle \nabla, \sigma, P \uplus \{a \#_? \bar{b}\} \rangle
\justifies \langle \nabla, \sigma, P \rangle
\using \rulefont{\#_? a\bar{b}}
\end{prooftree}
\\[3ex]
\begin{prooftree}
\langle \nabla, \sigma, P \uplus  \{a \#_? f\,t\} \rangle
\justifies \langle \nabla, \sigma, P \cup  \{a \#_? t\} \rangle
\using \rulefont{\#_? app}
\end{prooftree}
&
\begin{prooftree}
\langle \nabla, \sigma, P \uplus  \{a \#_? [a]t\} \rangle
\justifies \langle \nabla, \sigma, P\rangle
\using \rulefont{\#_? a[a]}
\end{prooftree}
\\[3ex]
\begin{prooftree}
\langle \nabla, \sigma, P \uplus  \{a \#_? [b]t\} \rangle
\justifies \langle \nabla, \sigma, P \cup  \{a \#_? t\}\rangle
\using \rulefont{\#_? a[b]}
\end{prooftree}
&
 \begin{prooftree}
 \langle \nabla, \sigma, P \uplus  \{a \#_? \pi . X\} \rangle
 \justifies \langle \{(\pi^{-1}\cdot a) \# X\} \cup \nabla, \sigma, P \rangle
 \using \rulefont{\#_? var}
 \end{prooftree}
 \\[3ex]
 \multicolumn{2}{c}
{
\begin{prooftree}
\langle \nabla, \sigma, P \uplus  \{a \#_? \langle s, t\rangle\} \rangle
\justifies 
\langle \nabla, \sigma, P \cup  \{a \#_? s, a \#_? t\} \rangle
\using \rulefont{\#_? pair}
\end{prooftree}
}
\end{array}
}\]\vspace{-3mm}
\caption{Reduction rules for freshness problems}
\label{fig:fresh_rules_algorithm}
\vspace{-4mm}
\end{figure}


\begin{example}\label{ex:fixpoint}
  Let $*$\footnote{Infix notation is adopted for commutative
    symbols: $s * t$ abbreviates $*\langle s,t\rangle$.} be a commutative function symbol.  Below, we show
  how the 
  problem
  $\mathcal{P}=\pair{\emptyset}{\{[e](a\,b).X * Y \approx_?  [f]
    (a\,c)(c\,d).X * Y\}}$
  reduces (via rules in Figs.~\ref{fig:eq_rules_algorithm} and
  ~\ref{fig:fresh_rules_algorithm}). Application of rule
  \rulefont{\approx_? C} gives two branches which reduce into
  two fixpoint problems: ${\cal Q}_1$ and ${\cal Q}_2$.  Highlighted
  terms show where the rules are applied.
For brevity, let
$\pi_1 = (a\,c)(c\,d)(e\,f)$, $\pi_2 = (a\,b)(e\,f)(c\,d)(a\,c)$, 
$\pi_3= (a\,c)(c\,d)(e\,f)(a\,b)$ and, $\sigma=\{X / (e\,f)(a\,b).Y\}$.
 
{\footnotesize
\noindent$\begin{array}{ll}
\triple{\emptyset}{id}{\{\colorbox{lightgray}{$[e](a\,b).X * Y \approx_? [f] (a\,c)(c\,d).X * Y$}\}}   &    \Rightarrow_{\rulefont{\approx_? [ab]}} \\
           \triple{\emptyset}{id}{\{\colorbox{lightgray}{$(a\,b).X * Y
  \approx_? \pi_1.X * (e\,f).Y$}, \,e \#_?  (a\,c)(c\,d).X * Y\}} &
   \Rightarrow_{\rulefont{\approx_? C}} 
\end{array}$

$\begin{array}{ll} 
\mbox{branch 1:}  & \triple{\emptyset}{id}{\{\colorbox{lightgray}{$(a\,b).X \approx_?  \pi_1.X$}, \,\colorbox{lightgray}{$Y \approx_? (e\,f).Y$},\,  
   \,e \#_?  (a\,c)(c\,d).X * Y \}} \\
    \Rightarrow_{\rulefont{\approx_? inv}}\!\!(2\times)  &
    \triple{\emptyset}{id}{\{(a\,b)[\pi_1]^{-1}.X  \!\approx_?\!  X, \, [(e\,f)]^{-1}.Y  \!\approx_?\!  Y,\, \colorbox{lightgray}{$e \#_?  (a\,c)(c\,d).X * Y$}\}}\\
  \Rightarrow_{\mbox{\scriptsize$\begin{array}{l}
                 \rulefont{\#_? app},\\
                  \rulefont{\#_? pair}
                 \end{array}$}} & \triple{\emptyset}{id}{\{\pi_2.X\approx_?X, (e\,f).Y \approx_? Y,\, 
                                                       \colorbox{lightgray}{$e\#_? (a\,c)(c\,d).X$},\, \colorbox{lightgray}{$e\#_? Y$} \}} \\
  \Rightarrow_{\rulefont{\#_? var}}\!\!(2\times) &
                                                   \triple{\{e\#X,e\#Y\}}{id}{\{\pi_2.X\approx_?X, (e\,f).Y \approx_? Y\}} ={\cal Q}_1
\end{array}$

$\begin{array}{ll} 
   \mbox{branch 2:} & 
   \triple{\emptyset}{id}{\{\colorbox{lightgray}{$(a\,b).X \approx_?  (e\,f).Y $}, \,Y \approx_? \pi_1.X,\,  
   \,e \#_?  (a\,c)(c\,d).X * Y \}} \\
    \Rightarrow_{\rulefont{\approx_? inst}} &
    \triple{\emptyset}{\sigma}{\{\colorbox{lightgray}{$Y \approx_? (a\,c)(c\,d)(e\,f)(e\,f)[(a\,b)]^{-1}.Y$},
                                                                            e \#_? \pi_1[(a\,b)]^{-1}.Y * Y\}}\\
  \Rightarrow_{\rulefont{\approx_? inv}} & \triple{\emptyset}{\sigma}{\{[(a\,c)(c\,d)(a\,b)]^{-1}.Y \approx_? Y,\, 
                                                       \colorbox{lightgray}{$e \#_? \pi_3.Y * Y$}\}} \\
  \Rightarrow_{\mbox{\scriptsize$\begin{array}{l}
                 \rulefont{\#_? app},\\
                  \rulefont{\#_? pair}
                 \end{array}$}} & \triple{\emptyset}{\sigma}{\{(a\,b)(c\,d)(a\,c).Y \approx_? Y,\, 
                                                       \colorbox{lightgray}{$e \#_? \pi_3.Y$},\, \colorbox{lightgray}{$e \#_?Y$}\}} \\
 \Rightarrow_{\rulefont{\#_? var}}\!\!(2\times)& \triple{\{e\#Y,f\#Y\}}{\sigma}{\{(a\,b)(c\,d)(a\,c).Y \approx_? Y\}} ={\cal Q}_2
       \end{array}
$
}
\end{example}

\begin{definition}[Set of $\onestepApp$ and $\onestepFresh$-normal forms]
We denote by ${\cal P}_{\onestepApp}$ (resp. ${\cal P}_{\onestepFresh}$) 
the set of normal forms of ${\cal P}$ with respect to $\onestepApp$ (resp. $\onestepFresh$).  
\end{definition}

\begin{definition}[Fail and success for $\onestepApp$] Let ${\cal
    P}$ be a triple, such that the rules in Fig.~\ref{fig:eq_rules_algorithm} give rise to
    a normal  form $\langle   \nabla, \sigma, P\rangle$. 
The rules in Fig.~\ref{fig:eq_rules_algorithm} are said to
\emph{fail} if  $P$ contains non fixpoint equations.  
Otherwise $\langle   \nabla, \sigma, P\rangle$ is called  a  \emph{successful} triple 
 regarding $\onestepApp$ (i.e., in a successful triple, $P$ consists only of fixpoint
equations and, possibly, freshness constraints).
\end{definition}

The rules in Fig. \ref{fig:fresh_rules_algorithm} will only be applied to
 successful triples regarding  $\onestepApp$.

\begin{definition}[Fail and success for $\onestepFresh$]\label{def:failsuccessRfresh} 
Let  ${\cal Q} = \langle \nabla, \sigma,
Q\rangle$ be a successful triple regarding $\onestepApp$, and ${\cal Q}' = \langle \nabla', \sigma, Q'\rangle$ its normal form via rules in Fig.~\ref{fig:fresh_rules_algorithm}, that is
${\cal Q}\onestepFresh^*{\cal Q}'$ and ${\cal Q}'$ is in
${\cal Q}_{\onestepFresh}$. If $Q'$ contains 
freshness constraints it is said that $\onestepFresh$ \emph{fails}
for $\cal Q$; otherwise, ${\cal Q}'$ will be called a \emph{successful}
triple for $\onestepFresh$.  
\end{definition}

\begin{remark}
Since  in a successful triple  regarding
$\onestepApp$, $\cal Q$, one has only fixpoint equations and
$\onestepFresh$ acts only over freshness constraints,  $Q'$ in the
definition above contains only fixpoint equations and freshness constraints. 
Also,  by a simple case analysis on $t$ one can
check that any triple with freshness constraints $a\#_?t$ is reducible
by $\onestepFresh$,  except when $t\equiv \bar{a}$. Hence
the freshness constraints in $Q'$ would be only of the
form  $a \#_? \bar{a}$. 
\end{remark}

The relation $\onestepApp$,
starts from a triple with the identity substitution and always
maintains a triple $\langle \nabla, \sigma', P' \rangle$ in which the
substitution $\sigma'$ does not affect the current problem 
$P'$. The same happens for  $\onestepFresh$ since the substitution
does not change with this relation. This motivates the
next definition and lemma.  

\begin{definition}[Valid triple]\label{ded:validtriple} ${\cal P} =
  \langle \nabla, \sigma, P \rangle$ {\em is valid} if 
  $im(\sigma)\cap dom(\sigma) = \emptyset$ and  $dom(\sigma)
  \cap Var(P) = \emptyset$.
\end{definition}

\begin{remark}
A substitution $\sigma$ in a valid triple $\cal P$ is \emph{idempotent},
that is, $\sigma\sigma = \sigma$. 
\end{remark}

Lemma \ref{lem:valid_preservation} is proved by case
  analysis on the rules used by $\onestepApp$ and $\onestepFresh$.
\begin{lemma}[Preservation of valid triples]\label{lem:valid_preservation}
If ${\cal P} = \langle \nabla, \sigma, P \rangle$ is valid and
${\cal P} \onestepApp \cup \onestepFresh \mathcal{P'}
=  \langle \nabla', \sigma', P' \rangle$, then ${\cal P}'$ is also
valid.
\end{lemma}
 
From now on,  we consider only valid triples. 

\begin{lemma}[Termination of $\onestepApp$ and $\onestepFresh$]\label{lem:eq_termination}
There is no infinite chain of reductions $\onestepApp$ (or $\onestepFresh$) starting from an 
arbitrary triple  ${\cal P} = \langle \nabla, \sigma, P\rangle$.
\end{lemma}

\begin{proof}
\begin{itemize}[leftmargin=*]
\item The proof for $\onestepApp$ is by well-founded induction on $\cal P$ using the measure $\|{\cal P}\| =\langle |Var(P_\approx)|, \|P\|, |P_\nfpe|\rangle$ with a lexicographic ordering, where
$\|P\| = \sum_{s \approx_? t \,\in\,P_\approx} |s| + |t|
   + \sum_{a \#_? u \in P_\#} |u|$. 
Note that this measure decreases after each step $\langle \nabla,
\sigma, P\rangle \onestepApp \langle \nabla, \sigma', P'\rangle$: 
for $\rulefont{\approx_?inst}$,  $|Var(P_\approx)|>|Var(P'_\approx)|$;
for $\rulefont{\approx_?refl}$, $\rulefont{\approx_? pair}$, $\rulefont{\approx_? app}$, 
 $\rulefont{\approx_?[aa]}$, $\rulefont{\approx_?[ab]}$ and
 $\rulefont{\approx_?C}$, $|Var(P_\approx)|$ $\geq$
 $|Var(P'_\approx)|$, but $\|P\|>\|P'\|$; 
and, for $\rulefont{\approx_?inv}$,  both $|Var(P_\approx)| =
 |Var(P'\approx)|$ and $\|P\|=\|P'\|$, but $|P_\nfpe|>|P'_\nfpe|$.
\item The proof for $\onestepFresh$ is by induction on $\cal P$ using as measure
$\|P_\#\|$. It can be checked that this measure decreases after
each step: $\langle \nabla,
\sigma, P\rangle \onestepFresh \langle \nabla, \sigma', P'\rangle$.  

\end{itemize}
\end{proof}

To solve a unification problem, $\langle\nabla, P\rangle$, one builds the derivation tree for
$\onestepApp$, labelling the root node with
$\langle\nabla, id, P\rangle$. This tree has leaves labelled with
$\onestepApp$-nf's that are either failing or successful
triples.  Then, the tree is extended by building 
 $\onestepFresh$-derivations starting from all successful leaves.  The
extended tree will include failing leaves and successful leaves. The
successful leaves will be labelled by triples ${\cal P}'$ in which the problem $P'$
consists only of fixpoint equations. 
Since $\onestepApp$ and $\onestepFresh$ are both terminating 
(Lemma \ref{lem:eq_termination}),
the process described above must be also terminating.

\begin{definition}[Derivation tree for $\langle\nabla, P\rangle$]\label{def:dertree}
A \emph{derivation tree} for the unification problem $\langle\nabla,
P\rangle$, denoted as ${\cal T}_{\langle\nabla, P\rangle}$, is a
tree with root label ${\cal P} = \langle\nabla, id,
P\rangle$ built in two stages:  
\begin{itemize}[leftmargin=*]
\item Initially, a tree is built, whose branches end in leaf nodes
  labelled with the triples in ${\cal P}_{\onestepApp}$. The labels in
  each path from the root to a leaf correspond to a 
  $\onestepApp$-derivation.     
\item Further, for each leaf labelled with a successful triple $\cal Q$ in ${\cal
    P}_{\onestepApp}$, the tree is
  extended with a path to a new leaf that is labelled with a
  $\bar{\cal Q} \in {\cal Q}_{\onestepFresh}$. The labels in the extended path
  from the node with label $\cal Q$ to the new leaf correspond to a
   $\onestepFresh$-derivation.    
\end{itemize}
\end{definition}

\begin{remark}
For  $\langle \nabla,
    P\rangle$,  all labels in the nodes of ${\cal T}_{\langle \nabla,
    P\rangle}$ are {\it valid} by Lemma \ref{lem:valid_preservation}. 
\end{remark}

The next lemma is proved by case analysis on elements of ${\cal P}_{\onestepApp}$ and ${\cal P}_{\onestepFresh}$.
\begin{lemma}[Characterisation of leaves of ${\cal T}_{\langle \nabla,
    P\rangle}$]\label{lem:char_leaves} 
Let $\langle \nabla,  P\rangle$ be a unification problem. 
If $\mathcal{P'} = \langle \nabla', \sigma', P'\rangle$
is the label of a leaf in ${\cal T}_{\langle \nabla, P\rangle}$,
then $P'$ can be partitioned as follows:  $P' = P''  \cup P_\bot$, where $P''$ is the set of all
fixpoint equations in $P'$ and $P_\bot = P' - P''$. 
If $P_\bot \neq \emptyset$ then
$\mathcal{U}_C(\mathcal{P'}) = \emptyset$. 
 \end{lemma}

The next definition is motivated by the previous characterisation of
the labels of leaves in derivation trees. 

\begin{definition}[Successful leaves]
Let $\langle \nabla,  P\rangle$ be a unification problem.  A leaf in
${\cal T}_{\langle \nabla, P\rangle}$ that is labelled with a triple
of the form $\mathcal{Q} = \langle \nabla', \sigma', Q\rangle$, where
$Q$ consists only of fixpoint equations, is called a \emph{successful
  leaf} of ${\cal T}_{\langle \nabla, P\rangle}$. In this case $\cal
Q$ is called a \emph{successful triple} of ${\cal T}_{\langle \nabla,
  P\rangle}$. The sets of successful leaves and triples  of ${\cal T}_{\langle
  \nabla, P\rangle}$ are denoted respectively by  $SL({\cal T}_{\langle \nabla,
  P\rangle})$ and $ST({\cal T}_{\langle \nabla,
  P\rangle})$.  
\end{definition}

The soundness theorem states that successful leaves of
  ${\cal T}_{\langle \nabla, P\rangle}$ produce \emph{correct} solutions.  The
  proof is by induction on the number of steps of $\onestepApp$ and
  $\onestepFresh$ and uses Lemma \ref{lem:char_leaves} and
  auxiliary results on  the \emph{preservation of solutions} by $\onestepApp$
  and $\onestepFresh$.  Proving preservation of solutions for rules
  $\rulefont{\approx_? [ab]}$ and $\rulefont{\approx_? inst}$ is not
  straightforward and uses Lemmas~\ref{lem:perm_subst_commutes} 
 \ref{lem:eq_reverse}, \ref{lem:fresh_preservation} and \ref{lem:equivariance}
 to check that the four conditions of Def. \ref{def:sol} are valid before, if 
 one supposes their validity after the rule application.

\begin{theorem}[Soundness of $\mathcal{T}_{\langle \nabla, P\rangle}$]\label{lem:T_correctness}
 $\mathcal{T}_{\langle \nabla, P\rangle}$ is correct, i.e., 
 if ${\cal P}' = \langle \nabla', \sigma, P' \rangle$ is the label of a leaf in $\mathcal{T}_{\langle \nabla, P\rangle}$, then
%
\it{1.} $\mathcal{U}_C(\mathcal{P'}) \subseteq \mathcal{U}_C(\langle
  \nabla, id, P\rangle)$, and
\it{2.} if $P'$ contains non fixpoint equations or freshness constraints then $\mathcal{U}_C(\mathcal{P'}) = \emptyset$.
\end{theorem}

The completeness theorem guarantees that the set of
  successful triples provides a complete set of
  solutions. Its proof uses case analysis on the rules of the relations 
$\onestepApp$ and $\onestepFresh$ by an argumentation similar to the
one used for  Theorem \ref{lem:T_correctness}.
For $\onestepFresh$ one has indeed equivalence:  ${\cal
  P}\onestepFresh {\cal P'}$, implies
$\mathcal{U}_C(\mathcal{P})=\mathcal{U}_C(\mathcal{P'})$. The same is true 
for all rules of the relation $\onestepApp$ except the branching
rule $\rulefont{\approx_?C}$, for which it is necessary to prove that all
solutions of a triple reduced by $\rulefont{\approx_?C}$  must belong
to the set of solutions of one of its sibling triples.

\begin{theorem}[Completeness of $\mathcal{T}_{\langle \nabla,
    P\rangle}$]\label{lem:T_completeness}
Let $\langle \nabla, P\rangle$ and  $\mathcal{T}_{\langle \nabla,
  P\rangle}$ be a unification problem and its derivation tree. Then 
${\cal U}_C(\langle \nabla, id,
P\rangle) = \bigcup_{{\cal Q}\in ST(\mathcal{T}_{\langle
    \nabla, P\rangle})}  {\cal U}_C({\cal Q})$. 
\end{theorem}

\begin{corollary}[Generality of successful triples]\label{cor:moregenerality}
Let ${\cal P}= \langle \nabla, P\rangle$ be a unification problem and $\langle
\Dta', \sigma'\rangle\in {\cal U}_C({\cal P})$. Then 
there exists a successful triple ${\cal Q} \in ST({\cal T}_{\langle \nabla,
  P\rangle})$ where ${\cal Q} = \langle \Dta, \sigma, Q\rangle$
such that $\langle \Dta', \sigma'\rangle\in {\cal U}_C({\cal Q})$,
and hence, $\Dta'\vdash \Dta\sigma'$ and there exists $\lambda$
such that $\Dta'\vdash \sigma\lambda\approx
\sigma'$.   
\end{corollary}
\begin{proof} By Theorem~\ref{lem:T_completeness},
  ${\cal U}_C({\cal P}) = \bigcup_{{\cal P}'\in ST({\cal T}_{\langle \nabla,
  P\rangle})} {\cal U}_C({\cal P}')$.  Then there exists  ${\cal Q} \in ST({\cal T}_{\langle \nabla,
  P\rangle})$ such that $\langle \Dta', \sigma'\rangle\in {\cal
  U}_C({\cal Q})$.  Suppose ${\cal Q} =   \langle \Dta, \sigma,
Q\rangle$. Then by the first and fourth conditions of the definition
of solution (Def. \ref{def:sol}) we have that $\Dta'\vdash
\Dta\sigma'$ and there exists $\lambda$ such that $\Dta'\vdash \sigma\lambda\approx 
\sigma'$. 
\end{proof}

\begin{remark}
The nominal C-unification problem is to decide, for a given
  $\cal P$, if ${\cal U}_C({\cal P})$ is non empty; that is,
  whether $\cal P$ has nominal C-unifiers. To prove that this problem
  is in NP, a non-deterministic procedure using the reduction rules in
  the same order as in Definition \ref{def:dertree} is designed. In
  this procedure, whenever rule $\rulefont{\approx_?C}$ applies, only
  one of the two possible branches is guessed. In this manner, if the
  derivation tree has a successful leaf, this procedure will guess a
  path to the successful leaf, answering positively to the decision
  problem.  According to the measures used in the proof of termination
  Lemma \ref{lem:eq_termination}, reduction with both the relations
  $\onestepApp$ and $\onestepFresh$ is polynomially bound, which
  implies that this non-deterministic procedure is polynomially bound.
  
  To prove NP-completeness, one can polynomially
  reduce the well-known NP-complete positive 1-in-3-SAT problem into
  nominal C-unification, as done in \cite{BaNi98} for the
  C-unification problem.  
    An instance of
  the positive 1-in-3-SAT problem consists of a set of clauses
  ${\cal C} = \{{\cal C}_i | 1\leq i\leq n\}$, where each ${\cal C}_i$
  is  a disjunction of three propositional variables, say
  ${\cal C}_i = p_i\vee q_i\vee r_i$. A solution of $\cal C$ is a
  valuation whit exactly one variable true in 
  each clause.  The proposed reduction of $\cal C$ into a nominal
  C-unification problem would require just a commutative function
  symbol, say $\oplus$, two atoms, say $a$ and $b$, a variable for each
  clause ${\cal C}_i$, say $Y_i$, and a variable for each
  propositional variable $p$ in $\cal C$, say $X_p$. Instantiating
  $X_p$ as $\barr{a}$ or $\barr{b}$, would be interpreted as
  evaluating $p$ as true or false, respectively.  Each clause
  ${\cal C}_i = p_i\vee q_i\vee r_i$ in $\cal C$ is translated into an
  equation $E_i$ of the form
  $((X_{p_i}\oplus X_{q_i})\oplus X_{r_i})\oplus Y_i \approx_?
  ((\barr{b}\oplus \barr{b}) \oplus \barr{a})\oplus ((\barr{b}\oplus
  \barr{a}) \oplus \barr{b})$.
  The nominal C-unification problem for $\cal C$ is given by
  ${\cal P_C}= \pair{\emptyset}{\{E_i | 1\leq i\leq n\}}$.
  Simplifying ${\cal P_C}$ would not introduce freshness constraints
  since the problem does not include abstractions.  Thus, to conclude
  it is only necessary to check that $\pair{\emptyset}{\sigma}$ is a
  solution for $\cal P_C$ if and only if $\sigma$ instantiates exactly
  one of the variables $X_{p_i}, X_{q_i}$ and $X_{r_i}$ in each
  equation with $\barr{a}$ and the other two with $\barr{b}$, which
  means that $\cal C$ has a solution.  
\end{remark}

\vspace{-3mm}
\section{Generation of solutions for successful leaves of $\mathcal{T}_{\langle \nabla, P\rangle}$}
\label{sec:complete_set_for_trees}
\vspace{-3mm}

To build solutions for a successful leaf
${\cal P} = \langle \nabla, \sigma, P\rangle$ in the derivation tree
of a given unification problem, we will select and combine solutions
generated for fixpoint equations $\pi.X \approx_? X$, for each
$X\in Var(P)$.  We introduce the notion of \emph{pseudo-cycle of a
  permutation}, in order to provide precise conditions to build terms
$t$ by combining the atoms in $dom(\pi)$, such that
$\pi\cdot t \appAC t$.  For convenience, we use the algebraic cycle
representation of permutations. Thus, instead of sequences of
swappings, permutations in nominal terms will be read as products of
disjoint cycles~\cite{Sagan2001}.

\begin{example}\label{ex:k-cycles}
(Continuing Example \ref{ex:fixpoint}) The permutations 
$(a\,b)\!\!::\!\!(e\,f)\!\!::\!\!(c\,d)\!\!::\!\!(a\,c)\!\!::\!nil$ and
$(a\,b)\!\!::\!\!(c\,d)\!\!::\!\!(a\,c)\!\!::\!nil$
are respectively represented as the product of permutation cycles 
$(a\,b\,c\,d)(e\,f)$ and $(a\,b\,c\,d)(e)(f)$.  
\end{example}

Permutation cycles of length one are omitted. In general the cyclic 
representation of a permutation
consists of the product of all its cycles.    

Let $\pi$ be a permutation with $dom(\pi) =n$.  Given
 $a\in dom(\pi)$ the elements of the 
 sequence  $a, \pi(a), \pi^2(a), \ldots$ cannot be all
 distinct. Taking the first $k\leq n$, such that
 $\pi^k(a)=a$, we have the $k$-cycle $(a \  \pi(a) \  \ldots
 \pi^{k-1}(a))$, where 
$\pi^{j+1}(a)$ is the \emph{successor} of
 $\pi^j(a)$. For the 4-cycle in the permutation $(a\, b\, c\, d)\,
 (e\, f)$, the 4-cycles generated by $a, b, c$ and $d$ are the same: 
 $(a\, b\, c\, d) = (b\, c\, d\, a) = (c\, d\, a\, b) = (d\, a\, b\, c)$.


Def.  \ref{def:pseudo_cycle} establishes the notion of a
\emph{pseudo-cycle w.r.t.\ a $k$-cycle $\kappa$}. Intuitively, given a
$k$-cycle $\kappa$ and a  commutative 
function symbol $*$, a pseudo-cycle w.r.t $\kappa$, $ (A_0 
\ldots A_l)$, is a cycle whose elements are either atom terms built
from the atoms in $\kappa$  or terms of the form $A_i' * A_j'$, for $A_i',
A_j'$ elements of  a pseudo-cycle w.r.t $\kappa$.  

\begin{definition}[Pseudo-cycle]\label{def:pseudo_cycle}
Let  $\kappa=(a_0 \ a_1 \ \ldots \ a_{k-1})$ be a $k$-cycle of a
permutation $\pi$.  A {\em
  pseudo-cycle} w.r.t. $\kappa$ is inductively defined as follows: 
\begin{enumerate}[leftmargin=*]
 \item  $\overline{\kappa} = (\overline{a_0}  \cdots
   \overline{a_{k-1}})$ is a {\em pseudo-cycle} w.r.t. $\kappa$, called
   \emph{trivial pseudo-cycle} of $\kappa$. 
 \item $\kappa' = (A_0\;  ...\,A_{k'-1})$ is a {\em pseudo-cycle} w.r.t. $\kappa$, if
 the following conditions are simultaneously satisfied:
 \begin{enumerate}[leftmargin=*]   
  \item each element of $\kappa'$ is of the form $B_i*
    B_j$, where $*$ is a commutative function symbol in the
    signature,  and $B_i, B_j$ are  different elements of
    $\kappa''$, a {\em 
      pseudo-cycle}  w.r.t. $\kappa$. $\kappa'$ will be called a  {\em
      first-instance pseudo-cycle} of $\kappa''$ w.r.t. $\kappa$.  
  \item  $A_i\not\!\approx_{\alpha,C} A_j$  for $i\neq j$, $0\leq i,j\leq k'-1$;
  \item for each $0 \leq i < k'-1$,    ${\kappa} \cdot A_i \appAC
    A_{(i+1) mod \, k'}$.
 \end{enumerate}
\end{enumerate} 
\end{definition}  

The \emph{length} of the pseudo-cycle  $\kappa$, denoted by $|\kappa|$, consists of the number of elements in $\kappa$. A pseudo-cycle of length one will be called \emph{unitary}.

\begin{example}\label{ex:pscycles01}

\begin{enumerate}[leftmargin=4mm,label=\Alph*] 
\item (Continuing Example \ref{ex:0})
 Considering $*$ a commutative function symbol in the signature,
 the unitary pseudo-cycles of $\kappa = (a\,b)$ generate infinite 
 independent solutions for the leaf $\triple{\emptyset}{id}{\{X \approx_?
   (a\,b).X\}}$.  
 Examples of these solutions are:
 $\pair{\emptyset}{\{X / \barr{a} * \barr{b}\}}$,
 $\pair{\emptyset}{\{X / (\barr{a} * \barr{a}) * (\barr{b} * \barr{b})\}}$,
 $\pair{\emptyset}{\{X / (\barr{a} * \barr{b}) * (\barr{a} * \barr{b})\}}$, 
 $\pair{\emptyset}{\{X / ((\barr{a} * \barr{a}) * \barr{a}) * ((\barr{b} * \barr{b}) * \barr{b})\}}$,
 $\pair{\emptyset}{\{X / (\barr{a} * (\barr{a} * \barr{a})) *
   (\barr{b} * (\barr{b} * \barr{b}))\}}$, etc.

\item (Continuing Examples~\ref{ex:fixpoint} and~\ref{ex:k-cycles}) 
In ${\cal Q}_1$ and ${\cal Q}_2$ we have the occurrences of the 4-cycle
$\kappa = (a\,b\,c\,d)$.
Suppose $*,\oplus,+$ are commutative operators in the signature. The following are pseudo-cycles w.r.t. $\kappa$:
$\overline{\kappa}=(\overline{a} \ \overline{b} \ \overline{c} \
  \overline{d} )$; 
$\kappa_1 = ( (\overline{a}* \overline{b}) \  (\overline{b}* \overline{c}) \ (\overline{c}* \overline{d}) \ (\overline{d}*\overline{a}))$;
$\kappa_2=( (\overline{a}\oplus \overline{c}) \  (\overline{b}\oplus
  \overline{d}))$; 
$\kappa_{11} =( ((\overline{a}* \overline{b}) +
  (\overline{b}* \overline{c}) ) \  ((\overline{b}* \overline{c}) +
  (\overline{c}*\overline{d}))   ((\overline{c}* \overline{d}) +
  (\overline{d}* \overline{a}) ) \  ((\overline{d}* \overline{a}) +
  (\overline{a}*\overline{b}))       )$;
$\kappa_{12} =( ((\overline{a}* \overline{b}) * (\overline{c}*
  \overline{d}) ) \  ((\overline{b}* \overline{c}) *
  (\overline{d}*\overline{a})))$;
$\kappa_{21} = ( ((\overline{a}\oplus \overline{c}) * (\overline{b}\oplus
  \overline{d}) ) )$;
$\kappa_{121} =( ((\overline{a}* \overline{b}) * (\overline{c}*
  \overline{d}) ) *  ((\overline{b}* \overline{c}) *
  (\overline{d}*\overline{a})))$.
$\kappa_1$ and $\kappa_2$ are first-instance pseudo-cycles of
$\overline{\kappa}$,  and $\kappa_{11}$ and $\kappa_{12}$ of
$\kappa_1$ and $\kappa_{21}$ of $\kappa_2$.  Notice that,
$|\barr{\kappa}|=|\kappa_1| = |\kappa_{11}|=4$, $|\kappa_{12}|=2$,
and $|\kappa_{21}| = |\kappa_{121} | = 1$.  Also,
$\kappa_{1}$ corresponds to $((\overline{a}*\overline{d})\, (\overline{b}*\overline{a})\,(\overline{c}*\overline{b})\,(\overline{d}*\overline{c}))$, a first-instance pseudo-cycle of
$\overline{\kappa}$. 

Finally, observe that for the elements  of the unitary
pseudo-cycles  $\kappa_{21}$ and $\kappa_{121}$, say $s=(\overline{a}\oplus \overline{c}) * (\overline{b}\oplus
  \overline{d})$ and $t=((\overline{a}* \overline{b}) * (\overline{c}*
  \overline{d}) ) *  ((\overline{b}* \overline{c}) *
  (\overline{d}*\overline{a}))$,
$\{X/s\}$ and $\{X/t\}$ (resp. $\{Y/s\}$ and $\{Y/t\}$) are solutions of the fixpoint equation
$(a\ b \ c \ d)(e \ f).X \approx_? X$ (resp. $(a\ b\ c\ d).Y \approx_?
Y$).  
\end{enumerate} 
\end{example}

Let $\kappa$ be a pseudo-cycle.  Notice that only item 2 of
  Def.~\ref{def:pseudo_cycle} may build a first-instance pseudo-cycle
  $\kappa'$ w.r.t. $\kappa$ with fewer elements.  If
  $|\kappa'| < |\kappa|$ then, due algebraic properties of cycles and
  commutativity of the operator applied ($*$), one must have that
  $|\kappa'| = |\kappa|/2$.  Thus, unitary pseudo-cycles can
  only be generated from cycles of length a power of two.  This is the
  intuition behind the next theorem, proved by induction on
  the size of the cycle $\kappa$.

\begin{theorem}\label{lem:power_2_solutions}
A pseudo-cycle $\kappa$ generates unitary pseudo-cycles  iff $|\kappa|$ is a power of two.
\end{theorem}

Notice that, according item 2.c of Def.~\ref{def:pseudo_cycle}, if
$\kappa' = (A_0 \ldots A_{k'-1})$ is a pseudo-cycle w.r.t. $\pi$ 
then $\pi \cdot A_{k'-1} \appAC A_0$; particularly, if $k' = 1$ then 
$\pi \cdot A_0 \appAC A_0$. 
Below, given ${\cal P}=\pair{\emptyset}{\{\pi . X \approx_? X\}}$ a fixpoint
equational problem, we call a \emph{combinatory solution} of ${\cal P}$, a
substitution $\{ X/ t\}$, such that $\pi \cdot t \approx_{C} t$, and
$t$ contains only atoms from $\pi$ and commutative function symbols,
built as unary  pseudo-cycles w.r.t. $\kappa$ a cycle in  $\pi$.  

The next theorem is proved by contradiction,  supposing that
  $\kappa$ has an odd factor and using Theorem
  \ref{lem:power_2_solutions}.

\begin{theorem}
Let ${\cal P}=\pair{\emptyset}{\{\pi.X \approx_? X\}}$ be a fixpoint problem. 
${\cal P}$ has a combinatory solution iff there exists a unitary pseudo-cycle $\kappa$ w.r.t.\ $\pi$.
\end{theorem}

\begin{remark}
Since one can generate infinitely many unitary pseudo-cycles from a
given $2^n$-cycle $\kappa$ in $\pi$, $n\in\mathbb{N}$, there exist
infinite independent solutions for the
fixpoint problem $\pair{\emptyset}{\{\pi.X\approx_? X\}}$.  
\end{remark}

\paragraph{General solutions for fixpoint problems.}
To compute the set of solutions for a fixpoint equation, we use a method described in~\cite{Ayala2017}, which is based on the computation of
unitary \emph{extended pseudo-cycles} (\epc).
We refer to~\cite{Ayala2017} for the definition of extended pseudo-cycles and an algorithm to enumerate all the solutions of a successful leaf in the derivation tree.

Pseudo-cycles are  built just from atom terms in $dom(\pi)$ and
commutative function symbols, while \epc's consider
all nominal syntactic elements including new variables, and also non
commutative function symbols.  The soundness and completeness of the generator of solutions described in~\cite{Ayala2017} 
relies on the properties of pseudo-cycles described above, in particular the fact that only unitary pseudo-cycles generate solutions.

\begin{example}
\begin{enumerate}[leftmargin=4mm,label=\Alph*] 
\item  (Continuing Example~\ref{ex:pscycles01} A) 
General solutions for the fixpoint  problem
$\pair{\emptyset}{\{(a\,b).X \approx_? X\}}$ include not only those
generated by pseudo-cycles, but also any other generated by 
unitary \epc's from $(a\,b)$, which include other syntactic elements to
the pseudo-cycle construction; for instance:
 $\pair{\emptyset}{\{X / h\,\barr{a} \!*\! h\,\barr{b}\}}$,
 $\pair{\{a,b\#Y}{\{X / (\pair{\barr{a}}{Y} \!*\! 
 \pair{\barr{a}}{Y}) \!\oplus\! (\pair{\barr{b}}{Y} \!*\! \pair{\barr{b}}{Y})\}}$,
 $\pair{\emptyset}{\{X / ([i]\barr{a} \!*\! [i]\barr{b}) \!\oplus\! ([i]\barr{a} \!*\! [i]\barr{b})\}}$,
 $\pair{\{a,b\#Z\}}{\{X / (\pair{Z}{\barr{a}} * (\barr{a} * \barr{a})) \oplus (\pair{Z}{\barr{b}} * (\barr{b} * \barr{b}))\}}$ and
 $\pair{\emptyset}{\{X / (([b]\barr{a} * [b]\barr{a}) * \barr{b}) \oplus (([a]\barr{b} * [a]\barr{b}) * \barr{a})\}}$.

\item (Continuing Example~\ref{ex:pscycles01} B)
Notice that constraints $e\#X$ and $e\#Y$ in 
${\cal Q}_1 =
\triple{\{e\#X,e\#Y\}}{id}{\{(a\,b\,c\,d)(e\,f).X\!\approx_?\!X, (e\,f).Y
  \!\approx_?\! Y\}}$ avoid any combinatory solution for atoms in 
$(e\,f)$, but there are unitary $\epc$'s w.r.t $(a\,b\,c\,d)$ that give 
solutions; for instance, \epc's 
$s_1 \!=\! (\overline{a}\!\oplus\! \overline{c}) \!*\! (\overline{b}\!\oplus\! \overline{d})$, 
$s_2 \!=\!  (h\,\overline{a}\!\oplus\! h\,\overline{c}) \!*\! (h\,\overline{b}\!\oplus\! h\,\overline{d})$
and 
$s_3 \!=\! ([i]\overline{a}\!\oplus\! [i]\overline{c}) \!*\!
([i]\overline{b}\!\oplus\! [i]\overline{d})$, generate solutions of the
form $\pair{\{e,f\#X,e,f\#Y\}}{\{X/s_i\}}$,  $i\!=\!1,2,3$;  the \epc\   
$v \!=\! (h\,\pair{\overline{a}}{Z}\!\oplus\! h\,\pair{\overline{c}}{Z}) \!*\!
(h\,\pair{\overline{b}}{Z}\!\oplus\! h\,\pair{\overline{d}}{Z})$ 
the solution  $\pair{\{e,f\#X,e,f\#Y, a,b,c,d,e,f\#Z\}}{\{X/v\}}$,
etc. 
\end{enumerate}
\end{example}

\vspace{-5mm}
\section{Conclusions and future work}
\label{sec:conclusions}
\vspace{-3mm}
 
A Coq formalisation of a sound and complete nominal
C-unification algorithm was obtained by combining  $\onestepApp$- and $\onestepFresh$-reduction. The algorithm builds finite derivation
trees such that the leaves provide a complete set of most general
unifiers consisting of freshness contexts, substitutions and fixpoint
equations. We have also introduced the notion of pseudo-cycle to eliminate fixpoint
equations by generating their possibly infinite set of solutions. The
generator is based on the analysis of permutation cycles, followed by a brute-force enumeration procedure. An implementantion of the algorithm, as well as its extension to deal with different equational theories is left as future work. 

\vspace{-3mm}

\end{document}